\documentclass[12pt, letterpaper]{article} 

\usepackage{amsthm}
\usepackage{multicol}
\usepackage{graphicx}
\usepackage{xcolor}
\usepackage{amsmath}
\usepackage{amssymb}
\usepackage{hyperref}
\usepackage{subcaption}
\usepackage{bbding}

\newtheorem{thm}{Theorem}

\newtheorem{rem}{Remark}

\def\pdf{\,p} 
\def\pmf{\,p} 
\def\pjf{\,p} 
\def\tp{q} 
\def\dop{\varsigma}

\title{{\small \copyright 2020. This manuscript version is made available under the CC-BY-NC-ND 4.0\\[-10pt] license http://creativecommons.org/licenses/by-nc-nd/4.0/} \\\vspace{2cm}~\\
Simultaneous Estimation of State and Packet-Loss Occurrences in Networked Control Systems} 

\author{A. Mohammadzadeh, 
B. Tavassoli,
B. Moaveni}

\date{}

\begin{document}
\maketitle

\begin{abstract}                          
Estimating the occurrence of packet losses in a networked control systems (NCS) can be used to improve the control performance and to detect failures or cyber-attacks. 
This study considers simultaneous estimation of the plant state and the packet loss occurrences at each time step. After formulation of the problem, two solutions are proposed. In the first one, an input-output representation of the NCS model is used to design a recursive filter for estimation of the packet loss occurrences. This estimation is then used for state estimation through Kalman filtering. In the second solution, a state space model of NCS is used to design an estimator for both the plant state and the packet loss occurrences which employs a Kalman filter. The effectiveness of the solutions is shown during an example and comparisons are made between the proposed solutions and another solution based on the interacting multiple model estimation method. 
\end{abstract}

\section{Introduction} \label{sec1}

The usage of communication networks for transfering data between sensors, plant, and controllers in networked control systems (NCSs) brings about several benefits such as reduction of wiring, flexibility, scalability and so on \cite{R:Zhang2017, 101}. 
However, these systems are also faced with communication effects such as data packet loss and delay. Packet loss occurrences can be represented as mode variables of the system that obey Markov chain models. The resulting system is a Markovian jump system (MJS) \cite{Costa2005,15,205}. 
It is much more complex to solve common control problems in the case of an NCS with a Markovian mode.
In many of the NCSs such as industrial control systems over wireless fieldbuses, packet loss is the primary issue. Awareness of the controller about occurrence of each packet loss is useful for detecting failures, cyber-attacks, or improving the control performance \cite{102}. While the networked controller is able to detect the loss of data packets that it should receive from sensors, it cannot directly detect loss of data packets that it sends to the actuators and an estimator of packet loss occurrence can be helpful in this regard. 

Due to the interaction between the state variables and the Markovian mode which stands for the packet loss occurrences, it is not easy to estimate the packet loss occurrences without having an estimation of the system's state. 
On the other hand, the ordinary Kalman filter for state estimation is not directly applicable to an MJS if the Markovian mode is not available. 
The existing methods for estimation of the MJS state mainly include the multiple model estimation techniques \cite{Li2005,Soken2020}. These methods are mainly applied to target tracking problems \cite{Li2005} while they also find other applications \cite{Shi2020}. Multiple model estimation can be extended to nonlinear MJSs \cite{Elenchezhiyan2015}.
When the mode variable is not a Markov chain, we may have to estimate the state and mode of a switching system \cite{Meyer2019,Domlan2007}.
Estimation of the mode variable is not an objective in multiple model estimation methods. Most of the other hybrid estimation approaches also focus only on state estimation. It is possible to design $H_{\infty}$ filters for estimating only the MJS state if the mode can be detected directly \cite{104} or without requiring information about mode and assuming it to be completely undetectable \cite{Souza2006a, Souza2006b} (at the cost of increased estimation error). 

It is also instructive to have a brief review of the extensive body of works on state estimation in the presence of packet losses. 
Modeling the successive packet losses as a Bernoulli process, Kalman filters are designed in \cite{R:Sinopoli2004,R:Mo2012,R:Liu2004} when the packet loss occurrences are detectable, and $H_{\infty}$ filters are designed in \cite{105,Li2010} when the packet loss occurrences are undetectable. Measurement quantization is additionally considered in \cite{106}. 
The conditions that ensure the stability of Kalman filter are studied for example in \cite{201,202}. While it is usual to drop the erroneous data packet due to noise, it is attempted to use the most of the noisy data in \cite{203}.

In this work, the focus is on estimation of the packet loss occurrences in an NCS which can be used for monitoring and control improvement as mentioned above.
Defining a Markovian mode variable which represents the occurrence of packet losses, the problem is formulated as an MJS filtering. The existing multiple model estimation methods can be applied for estimation of the mode using the model probabilities calculated in these methods, as it is done in the section for numerical example in this paper. However, these methods require to run multiple Kalman filters in parallel which can raise complexity issues when the number of mode values is large due to presence of several packet-based communication links in the NCS. 
The main contribution of this work is proposing two alternative methods in the form of algorithms without requiring multiple Kalman filters.
In the first method, an input-output representation of the NCS model is used to design a recursive estimator only for the mode variable. In this method, it is not necessary to run a Kalman filter in parallel with the mode estimator.
But, the estimated mode can be used by a Kalman filter for estimation of the state if required. 
In the second method, a state space model of the NCS is used for simultaneous estimation of the state and mode by designing an estimator which includes a single Kalman filter as a component. 
The effectiveness of both methods are verified and compared during an example. 

The paper is organized as follows. In Section 2, the required models for NCS with packet losses are obtained. The two proposed algorithms are developed in sections 3. The results are applied to the example problem in Section 4 and conclusions are made at the end of the paper.

{\it Notations}: The sets of real numbers and integer numbers are denoted by $\mathbb{R}$ and $\mathbb{Z}$ respectively. 
Given a continuous-valued random variable $v$ and a discrete-valued random variable $w$, the probability density function (PDF) for $v$ is denoted by $\pdf(v)$ and the probability that $w$ equals $\bar{w}$ is denoted as $\pmf(w=\bar{w})$. The joint probability distribution of $w$ and $v$ is defined as $\pjf(w=\bar{w},y)=\pdf(y|w=\bar{w})\pmf(w=\bar{w})$. The expected value of a random variable $v$ is denoted by $E[v]$. For a signal $y_k$ where $k$ is the time step, the history of signal is the sequence $Y_{k} = \{y_k: k\ge 0\}$. The Kronecker delta function is denoted by $\delta_{ij}$ for $i,j\in\mathbb{Z}$.

\section{Modeling} 
\label{se2}
Consider the following plant model in which $x_k\in \mathbb{R}^n$ stands for the state vector, $y_k\in \mathbb{R}^m$ denotes the observation vector, $w_k\in \mathbb{R}^n$ and $v_k\in \mathbb{R}^m$ are zero-mean white Gaussian uncorrelated random vectors with, $E[w_k w_i^T]=Q\delta_{ki},E[v_k v_i^T]=R\delta_{ki}$, $\hat{u}_k\in \mathbb{R}^r$ is the input signal and $A,B,C,D$ are system's matrices.
\begin{subequations}
\label{eq:1}
\begin{align}
x_{k+1} &=Ax_k+B\hat{u}_k+w_k \label{eq:1a} \\ 
y_k&=Cx_k+v_k  \label{eq:1b}
\end{align}
\end{subequations}

According to \cite{20} the system \eqref{eq:1} can be alternatively represented as 
\begin{align}
&y_k+\sum_{i=1}^{n}a_{i}y_{k-i}=\sum_{j=1}^{p}b_{j}\hat{u}_{k-j}+e_k+\sum_{m=1}^{h}c_{m}e_{k-m}\label{eq:2}
\end{align}
where ${e_k}$ is a linear combination of $v_k$ and $w_k$. Hence, $e_k$ is also a zero mean white Gaussian random vector such that $E[e_k e_i^T]$ has the form $\varLambda\delta_{ki}$.

By defining ${\dop}$ as the time shift operator such that
${\dop}^{-1}u_k=u_{k-1}$, the system equation \eqref{eq:2} is written as
\begin{align}
y_k=&-\widehat{A}({\dop}^{-1})y_k+\widehat{B}({\dop}^{-1})\hat{u}_k+\widehat{C}({\dop}^{-1})e_k\\
\widehat{A}({\dop}^{-1})&=a_{1}{\dop}^{-1}+\cdots+a_{n}{\dop}^{-n} \nonumber\\
\widehat{B}({\dop}^{-1})&=b_{1}{\dop}^{-1}+\cdots+b_{p}{\dop}^{-p} \nonumber\\
\widehat{C}({\dop}^{-1})&=1+c_{1}{\dop}^{-1}+\cdots+c_{h}{\dop}^{-h} \nonumber
\end{align}

\subsection{Packet losses}
For modeling packet losses in the input signal path, a new variable which shows the packet loss occurrence in $i$th input at the $k$th time step is defined as
\begin{equation}
\label{eq:alf}
\alpha_{i,k} = \begin{cases}0\qquad & \parbox{6cm}{if a packet loss occurs in the\\[-3pt] $i$th input link at the $k$th step}\\[6pt]
1 & \text{otherwise}\end{cases}
\end{equation}
and $\theta_{k}$ which is system's mode is defined as below where $r$ is the input vector dimensionality.
\begin{align}\label{amirr}
\theta_{k}=\begin{pmatrix}
\alpha_{1,k}\quad
\cdots\quad
\alpha_{r,k}
\end{pmatrix}^{T}
\end{align}

According to the above definition, $\theta_{k}$ belongs to a set of binary-valued vectors with $s=2^r$ elements denoted as $\Theta$. For simplicity, we represent this set as $\Theta=\{1..2^r\}$ by preserving the order of elements such that $\theta_k=1$ stands for $\alpha_{1,k}=\cdots=\alpha_{r,k}=0$ and $\theta_k=2^r$ stands for $\alpha_{1,k}=\cdots=\alpha_{r,k}=1$. 
Then, $\theta_{k}$ can be considered as a discrete-time Markov chain with transition probabilities 
\begin{equation}
\label{eq:tpr}
{\tp}_{ij}=\pmf(\theta_{k}=j \mid\theta_{k-1}=i)
\end{equation}

\begin{rem}
\label{remploss}
The transition probabilities ${\tp}_{ij}$ can be obtained empirically based on the above definition given a measured mode sequence. If the data packet transmissions in different communication links are not simultaneous or the communication mediums are separated, then $\alpha_{i,k}$ for $i\in\{1..r\}$ are independent binary-valued random variables. As a result, the transition probabilities in \eqref{eq:tpr} can be computed easily in terms of the indivdual packet loss probabilities that are much easier to be obtained empirically \cite{Parker2011}. It is also mentioned that, in some cases the packet loss probabilities are computable based on theoretical analysis of the underlying communication network \cite{DaSilva2019,Lee2002}.
\end{rem}

\subsubsection{Packet losses: zero strategy}
In the zero strategy \cite{Schenato2009}, the input will be replaced by zero if a packet loss occurs. By defining $\Gamma(\theta_{k})$ as
\begin{align} \label{eq:22}
\Gamma(\theta_{k})=\begin{pmatrix} 
\alpha_{1,k} &\cdots &0 \\
\vdots &\ddots &\vdots\\
0 &\cdots &\alpha_{r,k}
\end{pmatrix}
\end{align}
the lossy links at the input can be modeled as 
\begin{align}\label{anna1}
\hat{u}_k=\Gamma(\theta_{k}) u_k
\end{align}
where $u_k\in \mathbb{R}^r$ denotes the input signal sent via the communication link and $\hat{u}_k\in \mathbb{R}^r$ stands for input signal received at the actuator. 
Replacing $\hat{u}$ in the different representations of the plant model \eqref{eq:1} and \eqref{eq:2}, one can respectively obtain \eqref{statespace1} and \eqref{qw12} in the following.
\begin{subequations}
\label{statespace1}
\begin{align}
x_{k+1} &= Ax_k+B(\theta_k)u_k+w_k\\ 
y_k& = Cx_k+v_k \\
&B(\theta_k)=B\Gamma(\theta_k)\nonumber
\end{align}
\end{subequations}

\begin{align}
y_k+\sum_{i=1}^{n}a_{i}y_{k-i}=&\sum_{j=1}^pb_{j}\Gamma(\theta_{k-j})u_{k-j}+e_k 
+\sum_{m=1}^{h}c_{m}e_{k-m}\label{qw12}
\end{align}

The later equation can be written as
\begin{align}
y_k=&-\widehat{A}({\dop}^{-1})y_k+\widehat{B}({\dop}^{-1},\theta_{k-1},\theta_{k-2},\cdots,\theta_{k-p})u_k 
+\widehat{C}({\dop}^{-1})e_k\label{eq:packet1} \\
&\widehat{B}({\dop}^{-1},\theta_{k-1},\theta_{k-2},\cdots,\theta_{k-p})=\sum_{j=1}^pb_{j}\Gamma(\theta_{k-j}){\dop}^{-j} \nonumber
\end{align}

\subsubsection{Packet losses: hold strategy}
In the hold strategy \cite{Schenato2009}, the previous data will be used if a data packet is lost. The lossy link can be modeled as below instead of \eqref{anna1}.
\begin{align}\label{anna}
\hat{u}_k=\Gamma(\theta_{k})u_k+(I-\Gamma(\theta_{k}))\hat{u}_{k-1}\\\nonumber
\end{align}

To combine the above equation with the plant model \eqref{eq:1}, an augmented state vector is defined as
\begin{align}
\widehat{x}_k=\begin{pmatrix} 
x_k \\
\hat{u}_{k-1}
\end{pmatrix}
\end{align}

Then, the augmented plant model is obtained as
\begin{align}
\label{statespace2}
\widehat{x}_{k+1}&=A(\theta_k)\widehat{x}_k+B(\theta_k)u_k+\widehat{w}_k\\
y_k&=\begin{pmatrix} 
C \quad 0_{1\times r}
\end{pmatrix}\widehat{x}_k+v_k \nonumber\\
&A(\theta_k)=\begin{pmatrix} 
A&\quad\quad B(1-\Gamma(\theta_k)) \\
0&\quad\quad I-\Gamma(\theta_k)
\end{pmatrix} \nonumber\\
&B(\theta_k)=\begin{pmatrix} 
B\Gamma(\theta_k)\\
\Gamma(\theta_k)
\end{pmatrix} 
,\qquad \widehat{w}_k=\begin{pmatrix} 
w_k\\
0_{r\times 1}
\end{pmatrix} \nonumber
\end{align}

The input-output representation \eqref{eq:2} in combination with \eqref{anna} is also transformed to
\begin{align}
y_k+&\sum_{i=1}^na_{i}y_{k-i}=\sum_{j=1}^pb_{j}\Gamma(\theta_{k-j})u_{k-j}+ \nonumber\\
&\sum_{l=1}^pb_{l}\big(I-\Gamma(\theta_{k-l})\big)\hat{u}_{k-l-1}+e_k+\sum_{m=1}^he_{k-m}\nonumber
\end{align}

The above equation can be represented as
\begin{align}
y_k&=-\widehat{A}({\dop}^{-1})y_k+\widehat{B}({\dop}^{-1},\theta_{k-1},\theta_{k-2},\cdots,\theta_{k-p})u_k \nonumber\\ 
&~+\widehat{B}_{1}({\dop}^{-1},\theta_{k-1},\theta_{k-2},\cdots,\theta_{k-p})\hat{u}_k+C({\dop}^{-1})e_k\label{eq:packet2} \\
&\widehat{B}({\dop}^{-1},\theta_{k-1},\theta_{k-2},\cdots,\theta_{k-p})=\sum_{j=1}^pb_{j}\Gamma(\theta_{k-j}){\dop}^{-j} \nonumber\\
&\widehat{B}_{1}({\dop}^{-1},\theta_{k-1},\theta_{k-2},\cdots,\theta_{k-p})=\sum_{l=1}^pb_{j}(I-\Gamma(\theta_{k-l})){\dop}^{l-1} \nonumber
\end{align}

\section{Mode estimation} \label{sec3}
In this section, a recursive filter is designed to estimate the system mode $\theta_k$ from the output measurements $y_k$. For this purpose, a maximum likelihood estimation of $\theta_{k}$ is considered as 
\begin{align} 
\label{eq:123}
\widehat{\theta}_{k-1}=\text{argmax}_{j}~\pmf(\theta_{k-1}=j \mid Y_{k})
\end{align}

In the above equation, $\theta_{k-1}$ is estimated given $Y_k$, because $y(k)$ does not depend on $\theta_k$ according to \eqref{statespace1}. As usual, the mean value estimators and the estimation error covariances for the state variable $x$ in \eqref{statespace1} are defined as below.
\begin{subequations}
\label{eq:kd}
\begin{align}
\hat{x}_{k|k-1} &= E[x_k | Y_{k-1}] \\
\hat{x}_{k|k} &= E[x_k | Y_k] \\
P_{k|k-1} &= E[(x_k-\hat{x}_{k|k-1})(x_k-\hat{x}_{k|k-1})^T] \\
P_{k|k} &= E[(x_k-\hat{x}_{k|k})(x_k-\hat{x}_{k|k})^T]
\end{align}
\end{subequations}

Using the estimated mode $\widehat{\theta}_{k-1}$, the ordinary Kalman filter for time-varying systems given in the following can be used for state estimation by considering \eqref{statespace1} as a linear time varying system. 
\begin{subequations}
\label{eq:kf}
\begin{align}
\hat{x}_{k|k-1} &= A(\widehat{\theta}_{k-1}) \hat{x}_{k-1|k-1} + B(\widehat{\theta}_{k-1})u_{k-1} \label{eq:kf1}\\
P_{k|k-1} &= A(\widehat{\theta}_{k-1}) P_{k-1|k-1} A^T(\widehat{\theta}_{k-1}) + Q \label{eq:kf2}\\
\hat{x}_{k|k} &= \hat{x}_{k|k-1} + K_k (y_k - \hat{y}_k) \label{eq:kf3}\\
\hat{y}_k &= C \hat{x}_{k|k-1} \label{eq:kf4}\\
K_k &= P_{k|k-1} C^T (C P_{k|k-1} C^T + R)^{-1} \label{eq:kf5}\\
P_{k|k} &= P_{k|k-1} - K_k C P_{k|k-1} \label{eq:kf6}
\end{align}
\end{subequations}

The PDF of mode which is used for estimation in \eqref{eq:123} can be calculated recursively according to the following theorem.
\begin{thm}
The following recursive equation for $\pmf(\theta_{k-1}=j \mid Y_{k})$ holds.
\begin{align} \label{eq:491}
\pmf(\theta_{k-1}=j \mid Y_{k}) &= \frac{\pdf(y_k\mid \theta_{k-1}=j,Y_{k-1})E_{j,k}}{\sum_{h=1}^s \pdf(y_k \mid \theta_{k-1}=h,Y_{k-1}) E_{h,k}}\\ \nonumber
E_{h,k} &= \sum_{l=1}^s {\tp}_{lh} ~\pmf(\theta_{k-2}=l\mid Y_{k-1})\\\nonumber
\end{align} 
\end{thm}

\begin{proof}
Using the Bayes theorem, one can write the following equation.
\begin{align}
\pmf(\theta_{k-1}=j\mid Y_{k})=\frac{\pdf(\theta_{k-1}=j,y_k\mid Y_{k-1})}{\pdf(y_k\mid Y_{k-1})} \label{eq:y}
\end{align}
The numerator of the right hand side in \eqref{eq:y} can be written as
\begin{align}\nonumber
\pdf&(\theta_{k-1}=j,y_k\mid Y_{k-1})\\\nonumber
=&\pdf(y_k\mid \theta_{k-1}=j,Y_{k-1})\pmf(\theta_{k-1}=j\mid Y_{k-1})\\ \nonumber
=&\pdf(y_k\mid \theta_{k-1}=j,Y_{k-1})\Big(\sum_{i=1}^s \pmf(\theta_{k-1}=j,\theta_{k-2}=i\mid Y_{k-1})\Big)\\ \nonumber
=&\pdf(y_k\mid \theta_{k-1}=j,Y_{k-1}) 
\Big(\sum_{i=1}^s \pmf(\theta_{k-1}=j\mid \theta_{k-2}=i,Y_{k-1})\pmf(\theta_{k-2}=i\mid Y_{k-1})\Big)\\ \nonumber
\end{align}
Also, the denominator of the right hand side in \eqref{eq:y} can be written as
\begin{align}\nonumber
\pdf&(y_k\mid Y_{k-1})=\sum_{h=1}^s \pdf(y_k,\theta_{k-1}=h\mid Y_{k-1})\\ \nonumber
=&\sum_{h=1}^s \pdf(y_k\mid Y_{k-1},\theta_{k-1}=h)\pmf(\theta_{k-1}=h\mid Y_{k-1})\\ \nonumber
=&\sum_{h=1}^s \pdf(y_k\mid Y_{k-1},\theta_{k-1}=h)\Big(\sum_{l=1}^s \pmf(\theta_{k-1}=h,\theta_{k-2}=l\mid Y_{k-1})\Big)\\ \nonumber
=&\sum_{h=1}^s \pdf(y_k\mid Y_{k-1},\theta_{k-1}=h)\Big(\sum_{l=1}^s \pmf(\theta_{k-1}=h\mid \theta_{k-2}=l,Y_{k-1})\times\\[-6pt]
\nonumber
&\hspace{6cm} \pmf(\theta_{k-2}=l\mid Y_{k-1})\Big)
\end{align}
Due to the Markovian property of $\theta_{k}$ we have $\pmf(\theta_{k-1}=h\mid \theta_{k-2}=l,Y_{k-1})={\tp}_{lh}$. 
Then, by replacing the calculated numerator and the denominator of \eqref{eq:y}, the equation \eqref{eq:491} is resulted.
\end{proof} 

In order to use the above theorem, it is first needed to compute $\pdf(y_k\mid \theta_{k-1}=j,Y_{k-1})$. Combining the equations in \eqref{statespace1}, one can write
\begin{align*}
y_k& = C[A(\theta_{k-1})x_{k-1}+B(\theta_{k-1})u_{k-1}+w_{k-1}]+v_k 
\end{align*}

If $\theta_{k-1}$ is given, the right hand side of the above equation is composed of some Gaussian random variables and $x_{k-1}$. However, $x_{k-1}$ is a resultant of several random variables since the initial time. Therefore, the probability distribution of $x_{k-1}$ should not be far from the Gaussian distribution due to the central limit theorem. Hence, we assume that the following equations hold.

\begin{subequations}
\begin{align} \label{eq:y(t)}
\pdf(y_k &\mid \theta_{k-1}=j,Y_{k-1}) = \frac{\exp^{-\frac{1}{2}(y_k-\widehat{y}_{j,k})\Sigma^{-1}_{j,k}(y_k-\widehat{y}_{j,k})^{T}}}{\sqrt{(2\pi)^{m} |\Sigma_{j,k}|\,}}\\ 
\widehat{y}_{j,k} &= E(y_{k} \mid \theta_{k-1}=j,Y_{k-1}) \label{eq:ey} \\
\Sigma_{j,k} &= E[(y_{k}-\widehat{y}_{j,k})(y_{k}-\widehat{y}_{j,k})^T \mid \theta_{k-1}=j,Y_{k-1}] \label{eq:Sy}
\end{align} 
\end{subequations}

In the following, two approaches are proposed for calculating $\widehat{y}_{j,k}$ and $\Sigma_{j,k}$ in the the above equations.

\subsection{First approximation method}
The first approach is based on an approximate method for calculating $\widehat{y}_{j,k}$. The idea is to use the recursive equations of the system with packet losses. These equations are \eqref{eq:packet1} for the zero strategy and \eqref{eq:packet2} for the hold strategy. The expectation operation in \eqref{eq:ey} eliminates the noise terms and by replacing $\theta_{k-2},\cdots,\theta_{k-p}$ with their estimated values $\widehat{\theta}_{k-2},\cdots,\widehat{\theta}_{k-p}$ we have

\begin{subequations}
\begin{align} 
\widehat{y}_{j,k}&=\begin{cases} -\widehat{A}({\dop}^{-1})y_k+\widehat{B}({\dop}^{-1},j,\widehat{\theta}_{k-2},\cdots,\widehat{\theta}_{k-p})u_k & \text{zero strategy} \\[6pt] -\widehat{A}({\dop}^{-1})y_k+\widehat{B}({\dop}^{-1},j,\widehat{\theta}_{k-2},\cdots,\widehat{\theta}_{k-p})u_k+& \\[0pt] \hfil \widehat{B}_{1}({\dop}^{-1},j,\widehat{\theta}_{k-2},\cdots,\widehat{\theta}_{k-p})\hat{\hat{u}}_k & \text{hold strategy}\end{cases} \label{eq:ey1}\\[6pt]
\hat{\hat{u}}_{k}&=\Gamma(\widehat{\theta}_{k})u_{k}+\big(I-\Gamma(\widehat{\theta}_{k})\big)\hat{\hat{u}}_{k-1} \label{eq:ey1b}
\end{align}
\end{subequations}

The covariance matrix $\Sigma_{j,k}$ in \eqref{eq:Sy} can be also estimated by considering the independence of the noise terms at different time steps in \eqref{eq:packet1} and \eqref{eq:packet2} as 

\begin{equation}
\Sigma_{j,k} = (1+c_{1}^{2}+c_{2}^{2}+\cdots+c_{n}^{2})\varLambda \label{eq:Sy1}
\end{equation}

In the above equation $\varLambda$ is the covariance of the noise $e_k$ in equation \eqref{eq:2} and $c_{1},c_{2},\cdots,c_{n}$ are coefficients of the noise terms in that equation. 

The procedure for simultaneous estimation of mode and state based on the first approximation method can be represented as the following algorithm.

\vspace{6pt}
\noindent
{\bf Algorithm 1:} \\
{\bf Input}: The system model in\eqref{eq:packet1} for the zero strategy and \eqref{eq:packet2} for the \\[-2pt]hold strategy, the input $u_k$ at the $k$th step, the noise covariance matrices \\[-2pt]$R$ and $Q$, and the transition probabilities $\tp_{ij}$ defined in \eqref{eq:tpr}. \\
{\bf Initialization}: $\widehat{x}(0\mid 0)$, $P(0\mid 0)$, and $\pmf(\theta_{-1}=j\mid Y_{0})$ for $1\le j\le s$. \\
{\bf for} every time step $k$ {\bf do} \\ 
1. Calculate $\widehat{y}_{j,k}$ from \eqref{eq:ey1}. \\
2. Calculate $\Sigma_{j,k}$ from \eqref{eq:Sy1}. \\
3. Obtain $\pdf(y_k\mid\theta_{k-1}=j,Y_{k-1})$ from \eqref{eq:y(t)}. \\
4. Obtain $\pmf(\theta_{k-1}=j \mid Y_{k})$ from \eqref{eq:491}. \\
5. Obtain the mode estimation $\widehat{\theta}_{k-1}$ using \eqref{eq:123}. \\
6. Obtain the state estimation $\widehat{x}_{k\mid k}$ using the Kalman filter equations \\
\phantom{\hspace{3mm}} in \eqref{eq:kf} with $\theta_{k-1}$ set to $\widehat{\theta}_{k-1}$. \\
{\bf end}  

In the above algorithm, it is possible to estimate only the mode $\theta_{k-1}$ (without estimating the state $x_k$). For this purpose, it is only needed to eliminate the step 6 from the above algorithm.

\subsection{Second approximation method} 
In this part, $\widehat{y}_{j,k}$ in \eqref{eq:ey} is estimated using the state estimation obtained from the Kalman filter \eqref{eq:kf} as bellow
\begin{align}
\widehat{y}_{j,k} &= E(y_{k} \mid \theta_{k-1}=j,Y_{k-1}) \nonumber\\
&= E(Cx_k+v_k \mid \theta_{k-1}=j,Y_{k-1}) \nonumber\\
&= C E(A(\theta_{k-1})x_{k-1}+B(\theta_{k-1})u_{k-1}+w_{k-1} \mid \theta_{k-1}=j,Y_{k-1}) \nonumber\\
&= C A(j) E(x_{k-1} \mid Y_{k-1})+CB(j)u_{k-1} \quad\implies \nonumber\\
\widehat{y}_{j,k} &= C A(j) \hat{x}_{k-1|k-1}+CB(j)u_{k-1} \label{eq:ey2}
\end{align}

To compute the covariance matrix $\Sigma_{j,k}$ in \eqref{eq:Sy}, we first use \eqref{statespace1} to write the following equations given that $\theta_{k-1}=j$.
\begin{align*}
y_{k}-\widehat{y}_{j,k} =& C[A(j)x_{k-1}+B(j)u_{k-1}+w_{k-1}]+v_k \\[-4pt] 
&- (C A(j) \hat{x}_{k-1|k-1}+CB(j)u_{k-1}) \\
=&C A(j) (x_{k-1}-\hat{x}_{k-1|k-1})+C w_{k-1}+v_k \\
\end{align*}

Then, we can use \eqref{eq:kd} to write
\begin{align}
\Sigma_{j,k} &= E[(y_{k}-\widehat{y}_{j,k})(y_{k}-\widehat{y}_{j,k})^T \mid \theta_{k-1}=j,Y_{k-1}] \nonumber \\
&=C A(j) P_{k-1|k-1} A(j)^T C^T + C Q C^T + R \label{eq:Sy2}
\end{align}

With the above equations for $\widehat{y}_{j,k}$ and $\Sigma_{j,k}$, the Algorithm~1 can be modified as the following.

\vspace{6pt}
\noindent
{\bf Algorithm 2:} \\
This algorithm is the same as Algorithm 1, \\[-2pt]
except for steps 1 and 2 that are replaced by:\\
1. Calculate $\widehat{y}_{j,k}$ from \eqref{eq:ey2}. \\
2. Calculate $\Sigma_{j,k}$ from \eqref{eq:Sy2}. \\

Using the above algorithm, the mode $\theta_{k-1}$ and state $x_k$ must be estimated together and it is no longer possible to estimate the mode alone.

\begin{rem}
Algorithm 2 can be easily extended to the case in which the matrix $C$ in \eqref{eq:1b} depends on time $k$. For this purpose, it is only needed to replace $C$ by $C_k$ in \eqref{eq:kf}, \eqref{eq:ey2}, and \eqref{eq:Sy2}. This extension is useful when there are packet losses in the feedback path from the sensors to the controller. In this case, the matrix $C$ will depend on a new mode variable which is directly detectable and establishes a relationship between $y_k$ and the sample received by controller in the same way that $\theta_k$ establishes a relationship between $u_k$ and $\hat{u}_k$ in \eqref{anna}.	
\end{rem}

\begin{rem}
\label{ccmm}
The algorithms 1 and 2 have lower computational complexities compared with the multiple model estimation algorithms \cite{Li2005}. The reason is that the multiple model estimation algorithms generally need to run multiple Kalman filters in parallel. But, Algorithm~1 does not need a Kalman filter estimating only the mode (as explained after the algorithm) and the Algorithm~2 needs only a single Kalman filter. Excluding the Kalman filters, the remaining parts of the Algorithms 1, Algorithm 2, and the multiple model estimation algorithms have nearly the same computational loads that are less than the computational load of Kalman filtering.
\end{rem}

\section{Numerical example} \label{sec4}
In this section, the continuous stirred tank reactor (CSTR) process which is modeled in \cite[the 5th working point]{22} is considered for applying the results. Time discretization of the CSTR model with a sampling period of 0.25 sec results in the following state space equations.
\begin{align*}
x_{k+1}=&A_p x_k+B_p u_k\\ 
y_k=& x_k+v_k\\
A_p=& \begin{pmatrix} 
-0.8882&-0.0097\\
293.8556&2.2973\\
\end{pmatrix} , \quad
B_p = \begin{pmatrix} 
0.011&-0.0014\\
-0.3602&0.4732\\ 
\end{pmatrix} 
\end{align*}

The covariance matrix of the measurement noise $v_k$ is assumed to be equal to $R = 2.5\times 10^{-3} I$. 
The input-output representation of the system's model in \eqref{eq:2} can be also obtained as
\begin{align*}
y_k& = 1.4091y_{k-1}-0.8099y_{k-2}+\\
&~b_1 u_{k-1}+b_2 u_{t-2}+e_{k}-1.4091e_{k-1}+0.8099e_{k-2}\\
b_1 &= \begin{pmatrix}
 0.011&-0.0014\\
-0.3602&0.4732\\
\end{pmatrix} , \quad
b_2 = \begin{pmatrix}
-0.0218&-0.0014\\
2.9125&0.0089\\
\end{pmatrix}
\end{align*}
with $e_k = v_k$ which gives $\varLambda = E[e_k e_k^T] = R$. 

By using the hold strategy, the state-space representation of the system is in the form of \eqref{statespace2} and its equivalent input-output representation in \eqref{eq:packet2} can be obtained easily.

The above system model has $r=2$ and $\theta_k=(\alpha_{1,k}~~\alpha_{2,k})^T$ in \eqref{amirr} takes values from the set of four elements $\{(1~~1)^T$, $(1~~0)^T$, $(0~~1)^T$, $(0~~0)^T\}$ for different values of $\alpha_{1,k}$ and $\alpha_{2,k}$ defined in \eqref{eq:alf}.
It is assumed that $\alpha_{1,k}$ and $\alpha_{2,k}$ are independent binary-valued Markov chains with the following transition probability matrix (see Remark \ref{remploss}).
\begin{align*}
\begin{pmatrix} 
\pmf(\alpha_{i,k}=0|\alpha_{i,k-1}=0)\quad \pmf(\alpha_{i,k}=1|\alpha_{i,k-1}=0)\\
\pmf(\alpha_{i,k}=0|\alpha_{i,k-1}=1)\quad \pmf(\alpha_{i,k}=1|\alpha_{i,k-1}=1)\\
\end{pmatrix}
=
\begin{pmatrix} 
0.8\quad0.2\\
0.4\quad0.6\\
\end{pmatrix}
\end{align*} 

The independence of $\alpha_{1,k}$ and $\alpha_{2,k}$ can be used to write the following equation for calculating the transition probabilities of $\theta_{k}$.
\begin{align}
&\pmf(\theta_{k}=[i~~j] \mid\theta_{k-1}=[m~~n]) = \nonumber\\
&\quad\pmf(\alpha_{1,k}=i,\alpha_{2,k}=j\mid \alpha_{1,k-1}=m,\alpha_{1,k-1}=n )=\nonumber\\
&\quad \pmf(\alpha_{1,k}=i\mid \alpha_{1,k-1}=m )\pmf(\alpha_{2,k}=j\mid \alpha_{1,k-1}=n ) \label{eq:ex1}
\end{align}

According to the explanations underneath the Equation \eqref{amirr}, the set of values for $\theta_k$ is represented as $\{1,2,3,4\}$ for simplicity. More precisely, the mode $\theta_k$ is interpreted according to the Table~\ref{table2}.

\begin{table}[!b]
\small
\caption{Interpretation of the mode $\theta_k$ in terms of the packet loss occurrences for the example system.}
\label{table2}
\begin{center}
\begin{tabular}{|c|c|c|}
\cline{2-3}
\multicolumn{1}{c|}{} & \multicolumn{2}{|c|}{Packet loss occurrence} \\
\hline
Mode & ~first input~ & second input \\ 
\hline
1 & delivery & delivery \\ 
\hline
2 & delivery & loss \\
\hline
3 & loss & delivery \\
\hline
4 & loss & loss \\
\hline
\end{tabular}
\end{center}
\end{table} 

Then, \eqref{eq:ex1} can be used to obtain the transition probability matrix for $\theta_{k}$ with entries in \eqref{eq:tpr} as
 \begin{align*}
\begin{pmatrix} 
{\tp}_{11}&{\tp}_{12}&{\tp}_{13}&{\tp}_{14}\\
{\tp}_{21}&{\tp}_{22}&{\tp}_{23}&{\tp}_{24}\\
{\tp}_{31}&{\tp}_{32}&{\tp}_{33}&{\tp}_{34}\\
{\tp}_{41}&{\tp}_{42}&{\tp}_{43}&{\tp}_{44}\\
\end{pmatrix}
=
\begin{pmatrix} 
0.64&0.16&0.16&0.04\\
0.32&0.48&0.08&0.12\\
0.32&0.08&0.48&0.12\\
0.16&0.24&0.24&0.36\\
\end{pmatrix}.
\end{align*} 

Each element of the input $u_k$ is assumed to be a zero mean white noise with a standard deviation of 10. 
The initial state is set as $x_0=(1~~1~~1~~1)^T$.
The Kalman filter and mode estimation algorithms are also initialized as 
\begin{align*}
&\widehat{x}_{0\mid 0}=(0\quad 0 \quad 0 \quad 0)^{T},&& P_{0\mid 0} = 0.1 I_{4\times 4}\\
& \pmf(\theta_{-1}=i\mid Y_{0})=0.25 && i\in\{1..4\}.
\end{align*}

The above information provides the required data for applying the algorithms 1 and 2 to the CSTR example. 

The simulation results for applying the Algorithm 1 over 100 simulation steps are presented in Fig.~\ref{fig1}. The actual mode and its estimated value are shown in Fig.~\ref{fig1a}. 
The two plots coincide except at a few time steps at which the incorrectly estimated mode is marked by a {~\tiny\XSolidBold~} sign. The remaining subfigures in Fig. \ref{fig1} show the state variables, and the state estimation error in the Kalman filter. 
The simulation results for applying the Algorithm 2 are also presented in Fig. \ref{fig2} which shows the same set of information with the same format. 

\begin{figure}
\centering
\begin{subfigure}[b]{\textwidth}
\centering
\hfill\includegraphics[width=0.98\textwidth]{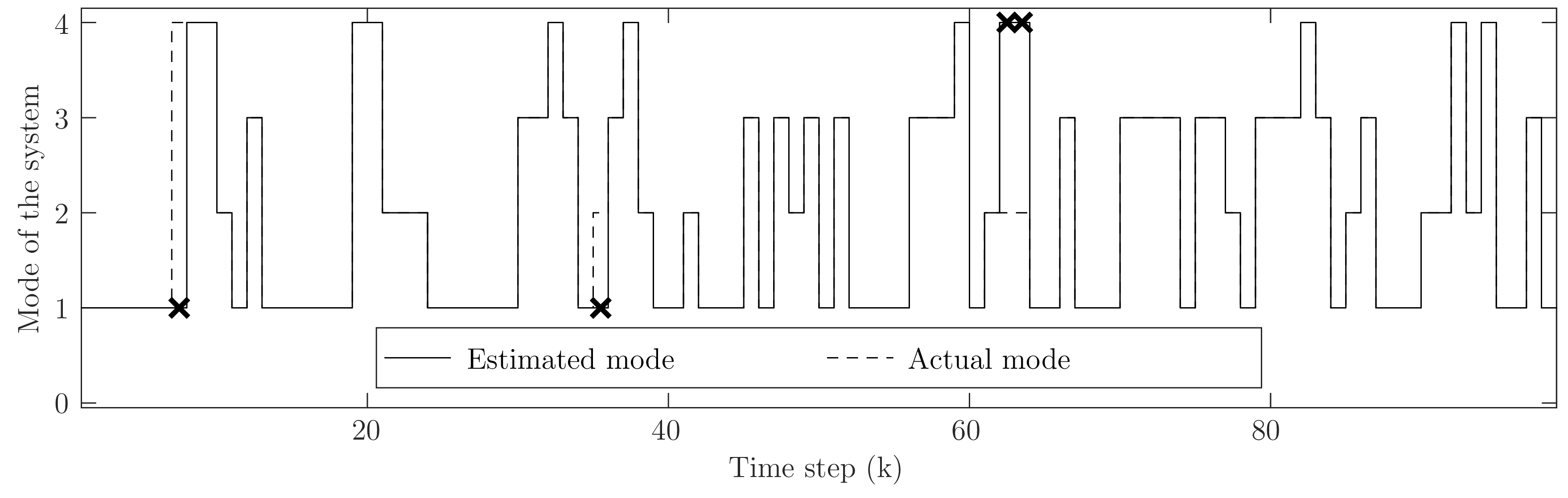}
\caption{The actual mode of system and its estimation using Algorithm 1.}
\label{fig1a}
\end{subfigure}
\par\vspace{12pt}
\begin{subfigure}[b]{2.2in}
\centering
\includegraphics[width=2.2in]{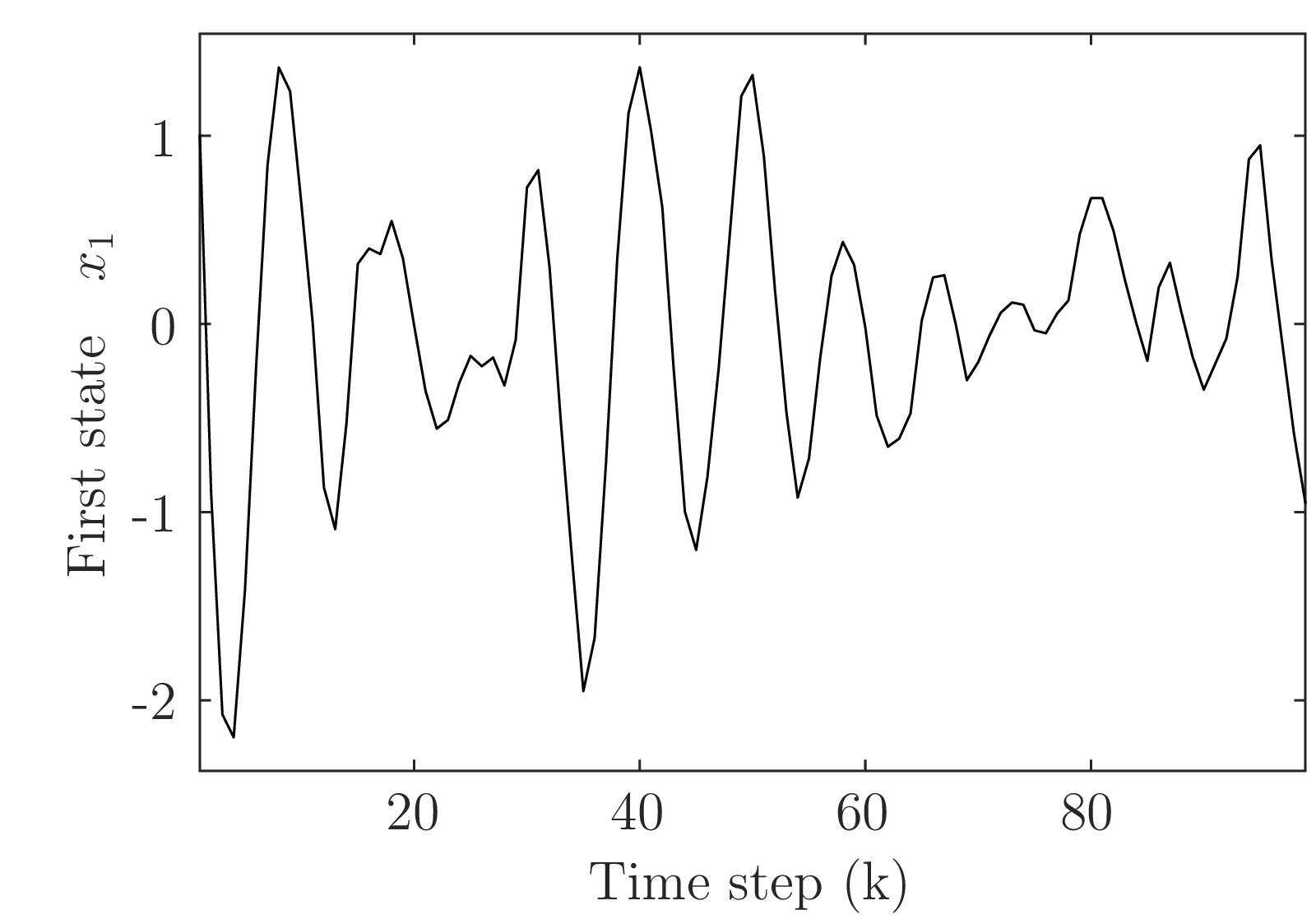}
\caption{First state $x_{1,k}$.}
\label{fig1b}
\end{subfigure}
\hfill
\begin{subfigure}[b]{2.2in}
\centering
\includegraphics[width=2.2in]{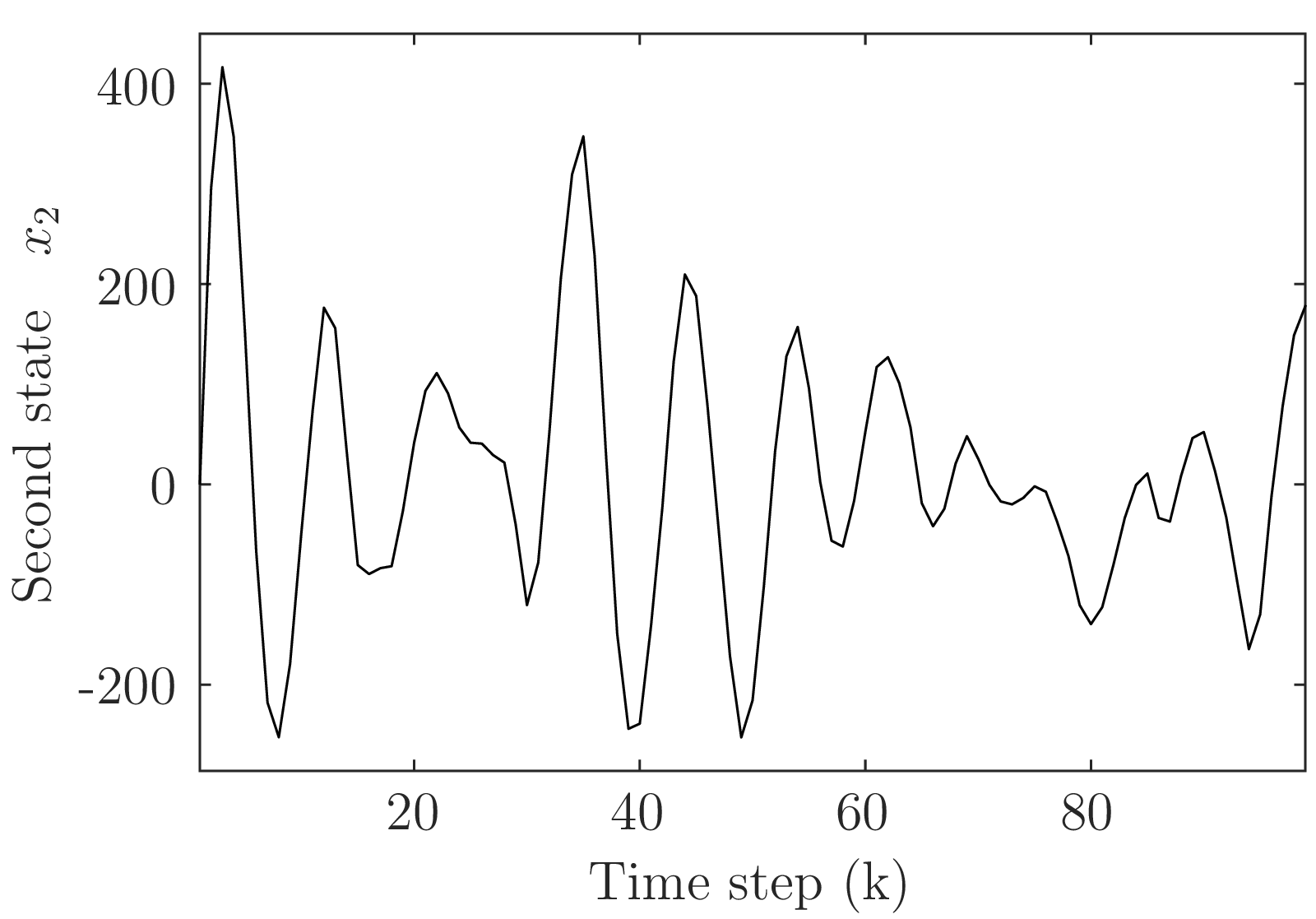}
\caption{Second state $x_{2,k}$.}
\label{fig1c}
\end{subfigure}
\par\vspace{12pt}
\begin{subfigure}[b]{2.2in}
\centering
\includegraphics[width=2.2in]{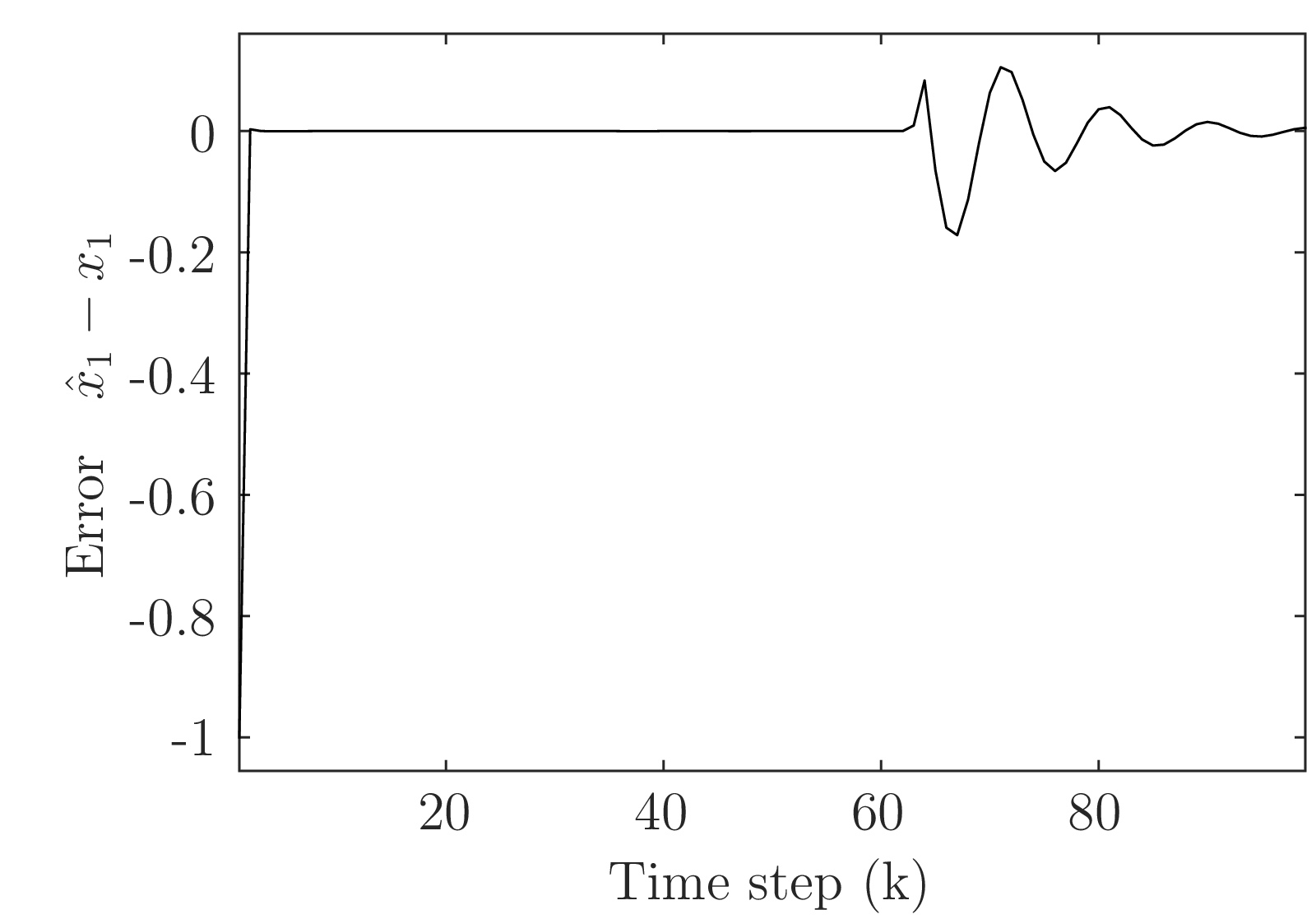}
\caption{Estimation error $x_{1,k}-\widehat{x}_{1,k}$.}\label{fig1d}
\end{subfigure}
\hfill
\begin{subfigure}[b]{2.2in}
\centering
\includegraphics[width=2.2in]{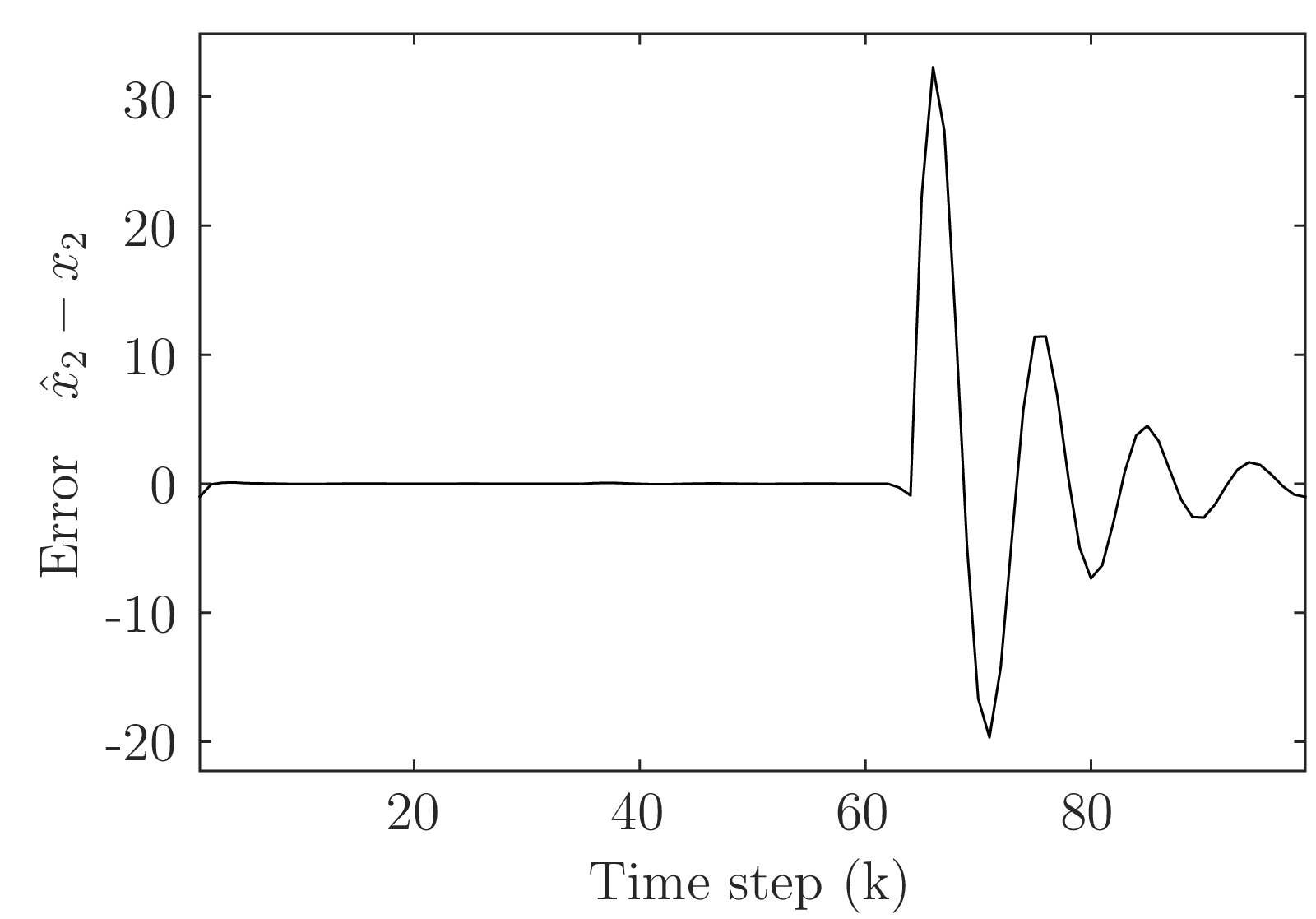}
\caption{Estimation error $x_{2}(k)-\widehat{x}_{2,k}$.}\label{fig1e}
\end{subfigure}
\caption{Simulation results for Algorithm 1.}
\label{fig1}
\end{figure}

\begin{figure}
\centering
\begin{subfigure}[b]{\textwidth}
\centering
\hfill\includegraphics[width=0.98\textwidth]{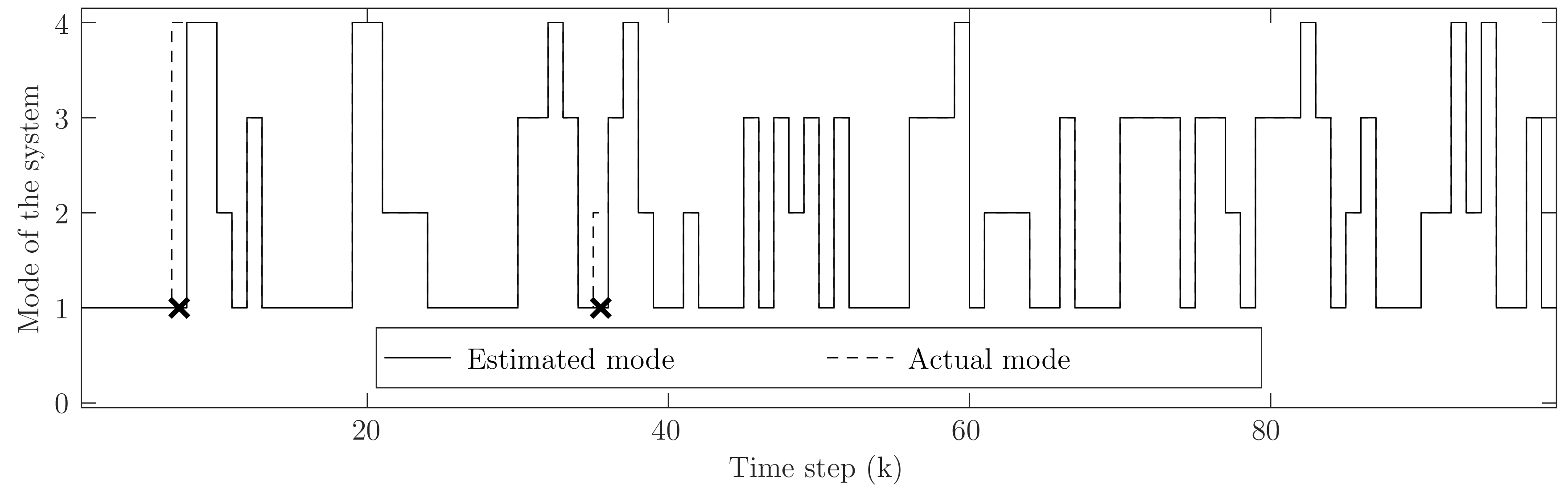}
\caption{The actual mode of system and its estimation using Algorithm 2.}
\label{fig2a}
\end{subfigure}
\par\vspace{12pt}
\begin{subfigure}[b]{2.2in}
\centering
\includegraphics[width=2.2in]{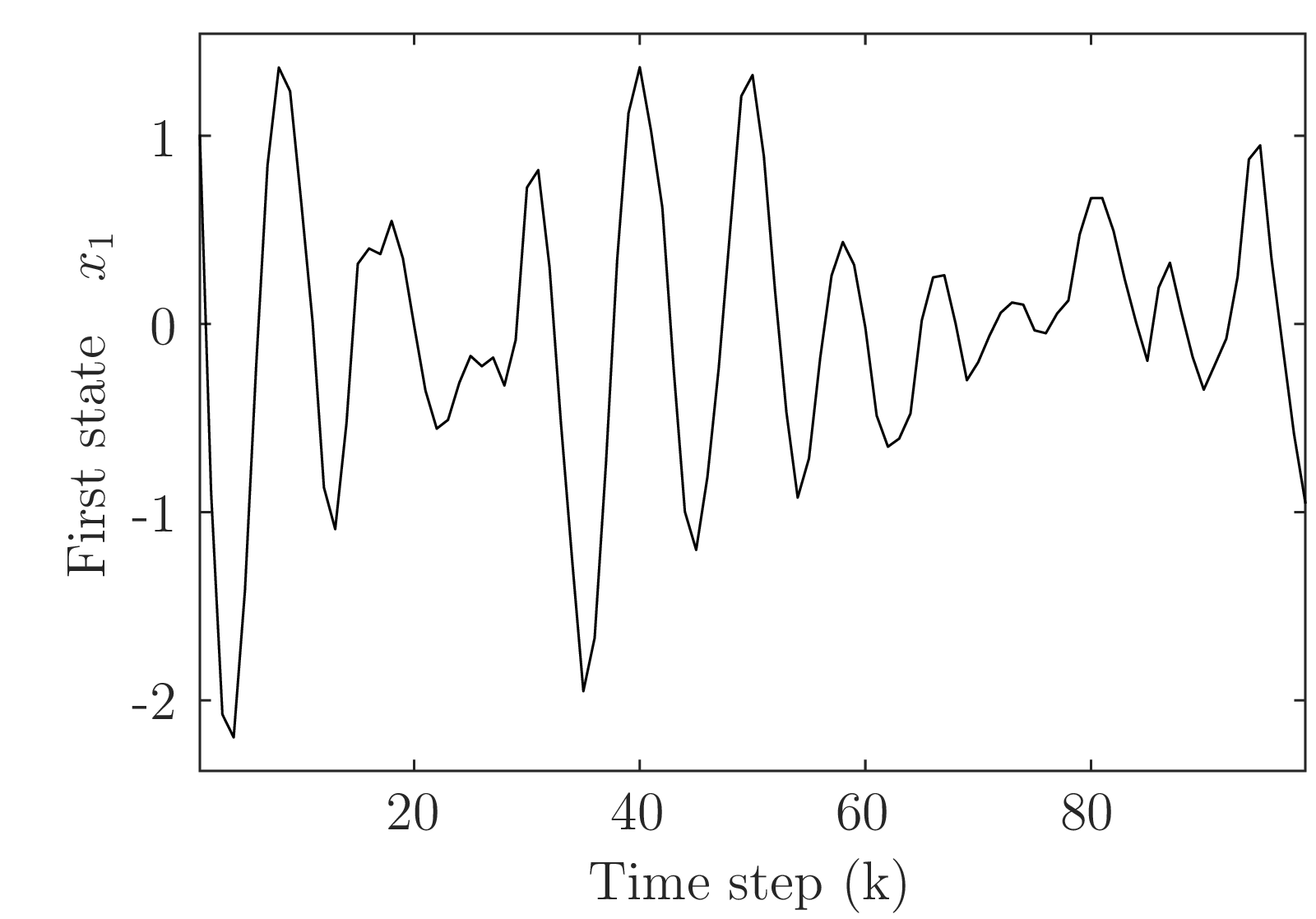}
\caption{First state $x_{1,k}$.}
\label{fig2b}
\end{subfigure}
\hfill
\begin{subfigure}[b]{2.2in}
\centering
\includegraphics[width=2.2in]{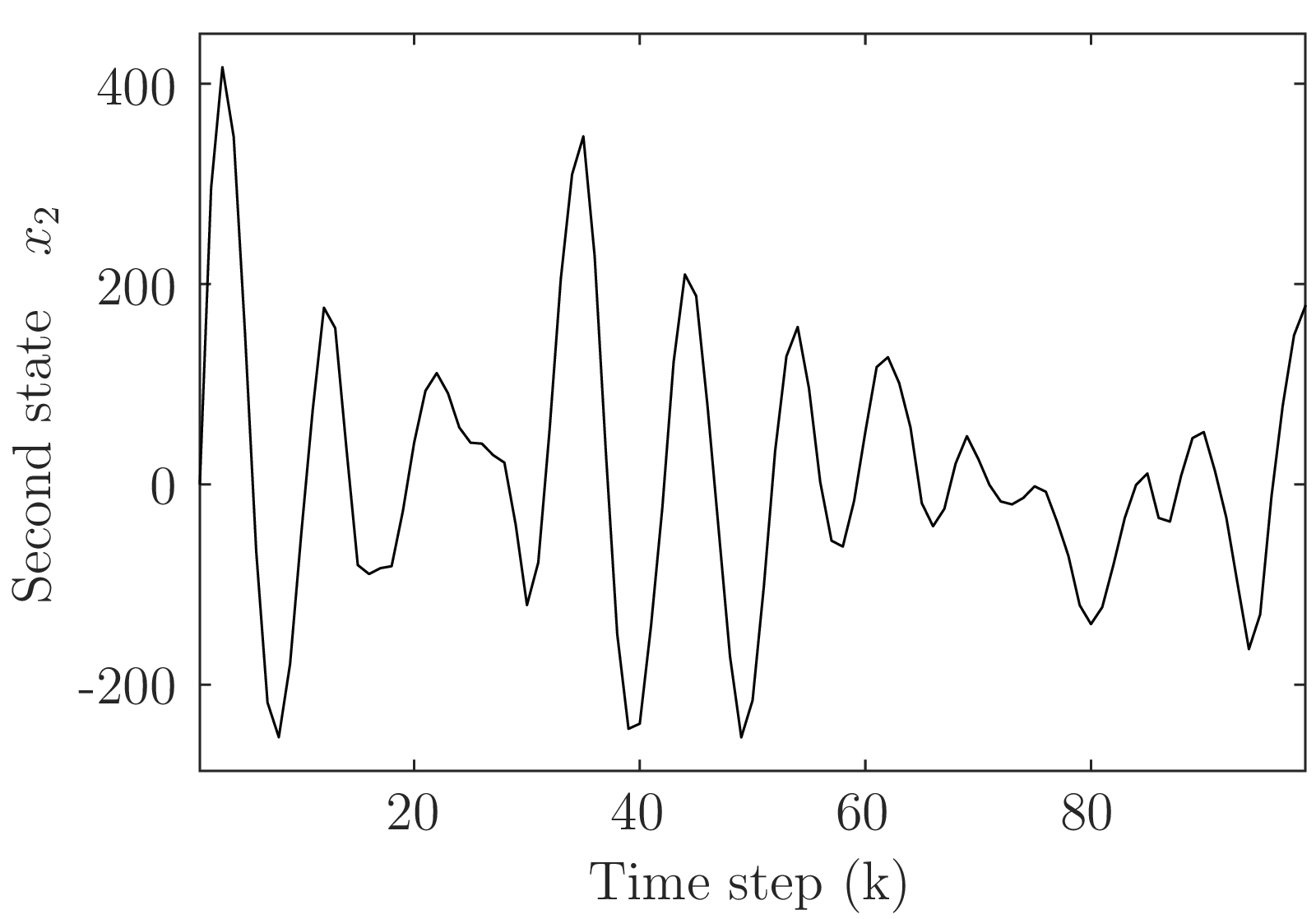}
\caption{Second state $x_{2,k}$.}
\label{fig2c}
\end{subfigure}
\par\vspace{12pt}
\begin{subfigure}[b]{2.2in}
\centering
\includegraphics[width=2.2in]{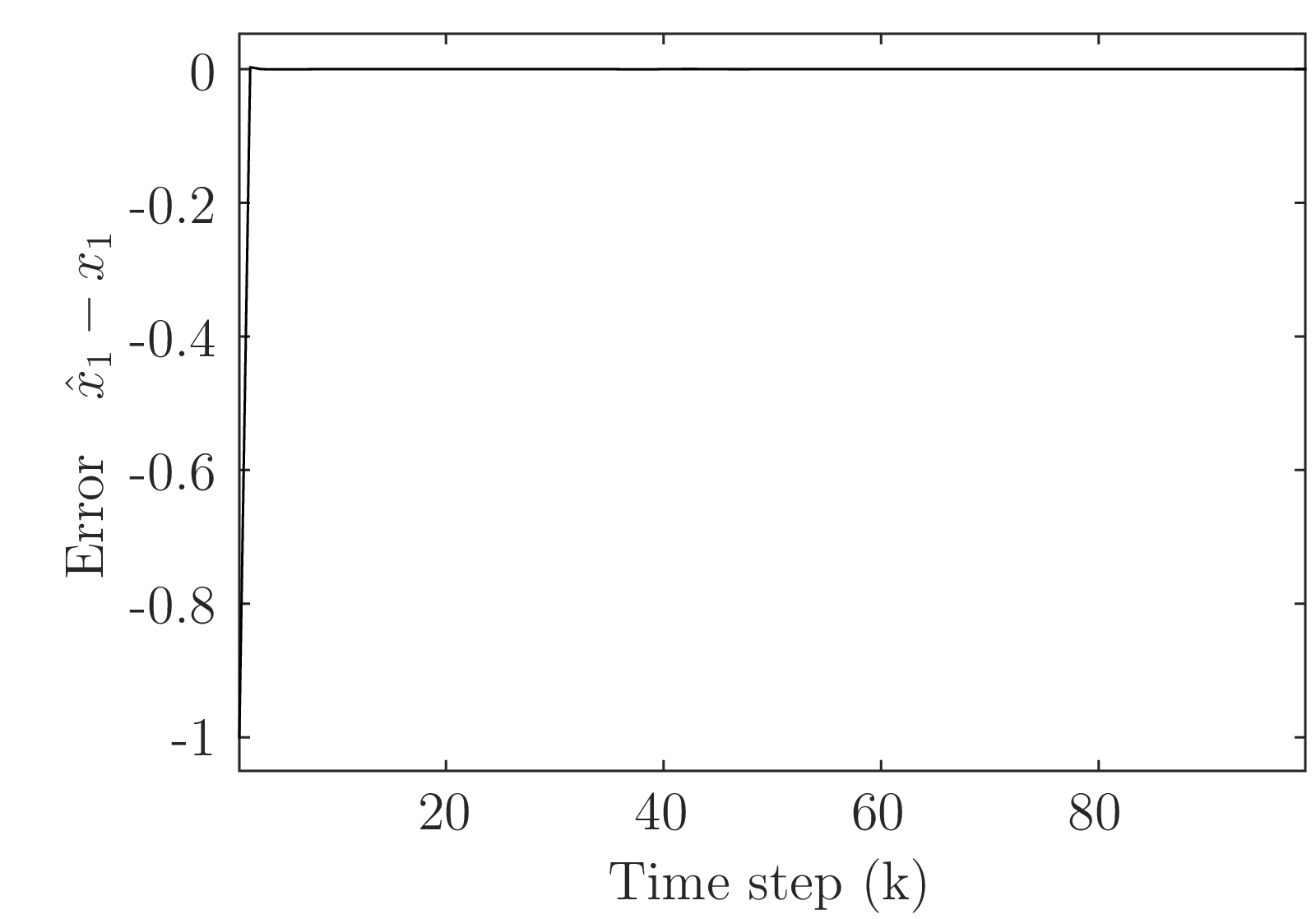}
\caption{Estimation error $x_{1,k}-\widehat{x}_{1,k}$.}\label{fig2d}
\end{subfigure}
\hfill
\begin{subfigure}[b]{2.2in}
\centering
\includegraphics[width=2.2in]{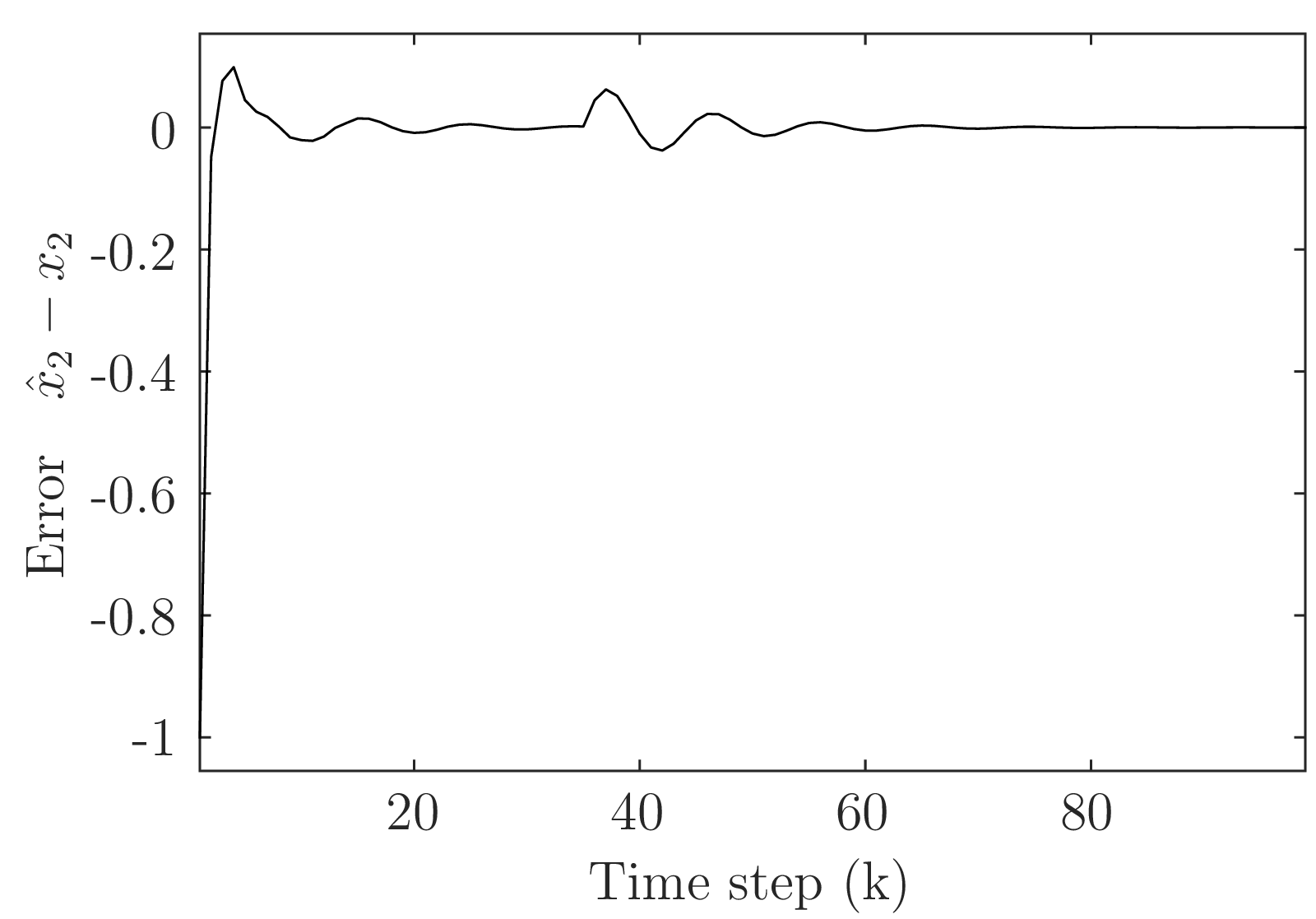}
\caption{Estimation error $x_{2,k}-\widehat{x}_{2,k}$.}\label{fig2e}
\end{subfigure}
\caption{Simulation results for Algorithm 2.}
\label{fig2}
\end{figure}

As mentioned in the Remark~\ref{ccmm}, it is possible to apply the multiple model estimation methods for simultaneous estimation of the mode $\theta_k$ and the state $x_k$. For this purpose, $\pmf(\theta_{k-1}=j \mid Y_{k})$ in \eqref{eq:123} is replaced with the model probabilities in the multiple model estimation methods at each time step. The results obtained by apply the interacting multiple model estimation method (IMM) as described in \cite{Li2005} are plotted in Fig.~\ref{fig3} which has the same format as the previous two figures.
According to the Figs. \ref{fig1} through \ref{fig3}, both Algorithm~1 and Algorithm~2 show acceptable performance compared to the IMM method which has a higher computational load due to running multiple Kalman filters in parallel.  
It is noticeable that there are time steps around which the mode estimation errors occur in all of the three methods. The reason is large noise amplitudes near these time steps.
 
\begin{figure}
\centering
\begin{subfigure}[b]{\textwidth}
\centering
\hfill\includegraphics[width=0.98\textwidth]{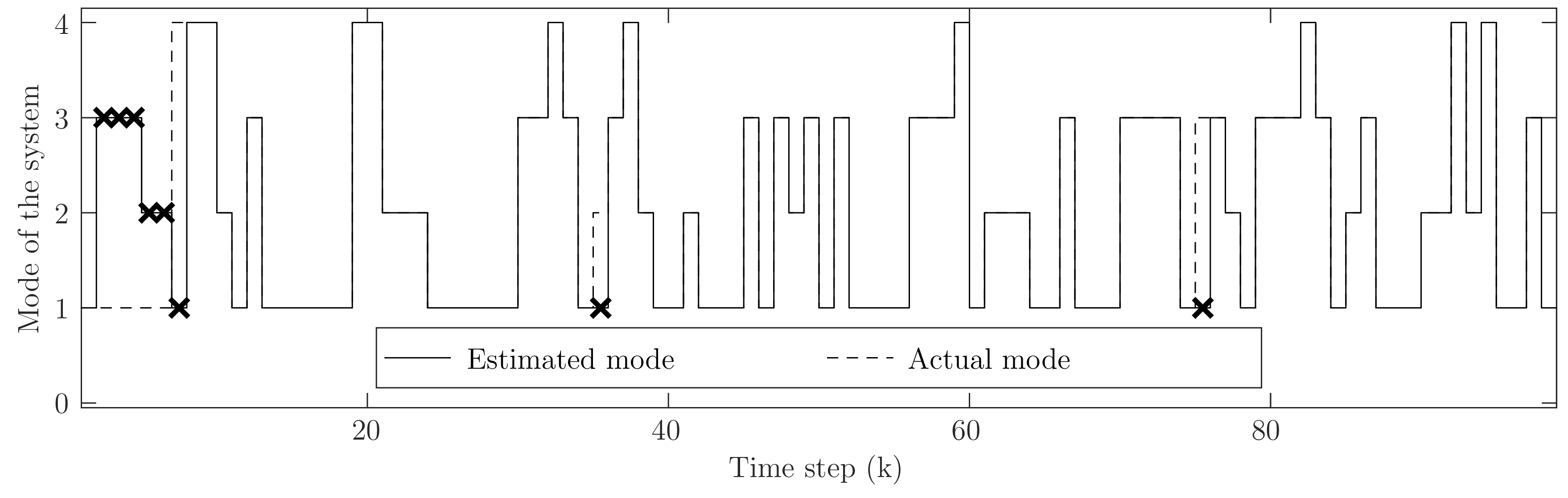}
\caption{The actual mode of system and its estimation using the IMM Algorithm.}
\label{fig3a}
\end{subfigure}
\par\vspace{12pt}
\begin{subfigure}[b]{2.2in}
\centering
\includegraphics[width=2.2in]{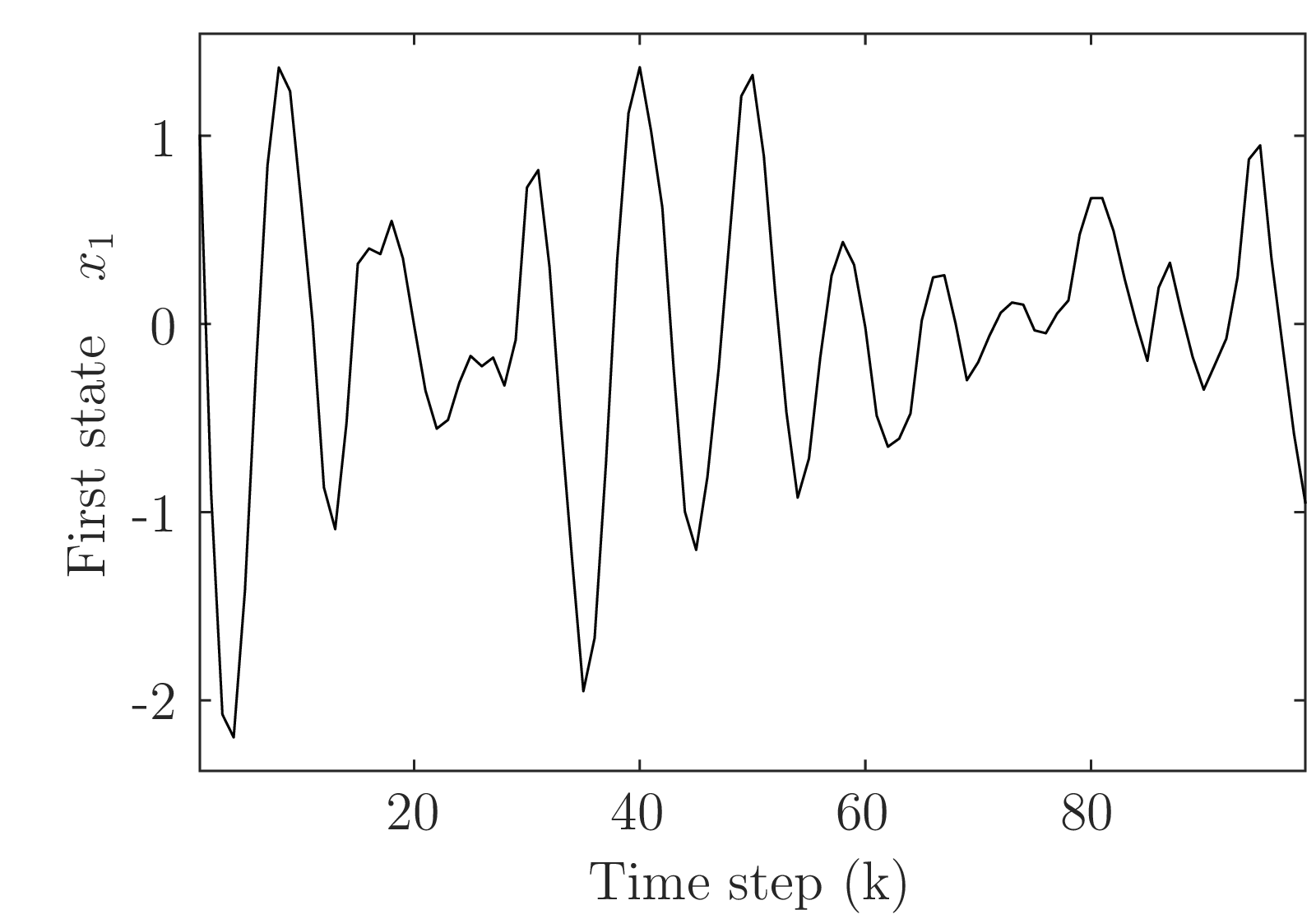}
\caption{First state $x_{1,k}$.}
\label{fig3b}
\end{subfigure}
\hfill
\begin{subfigure}[b]{2.2in}
\centering
\includegraphics[width=2.2in]{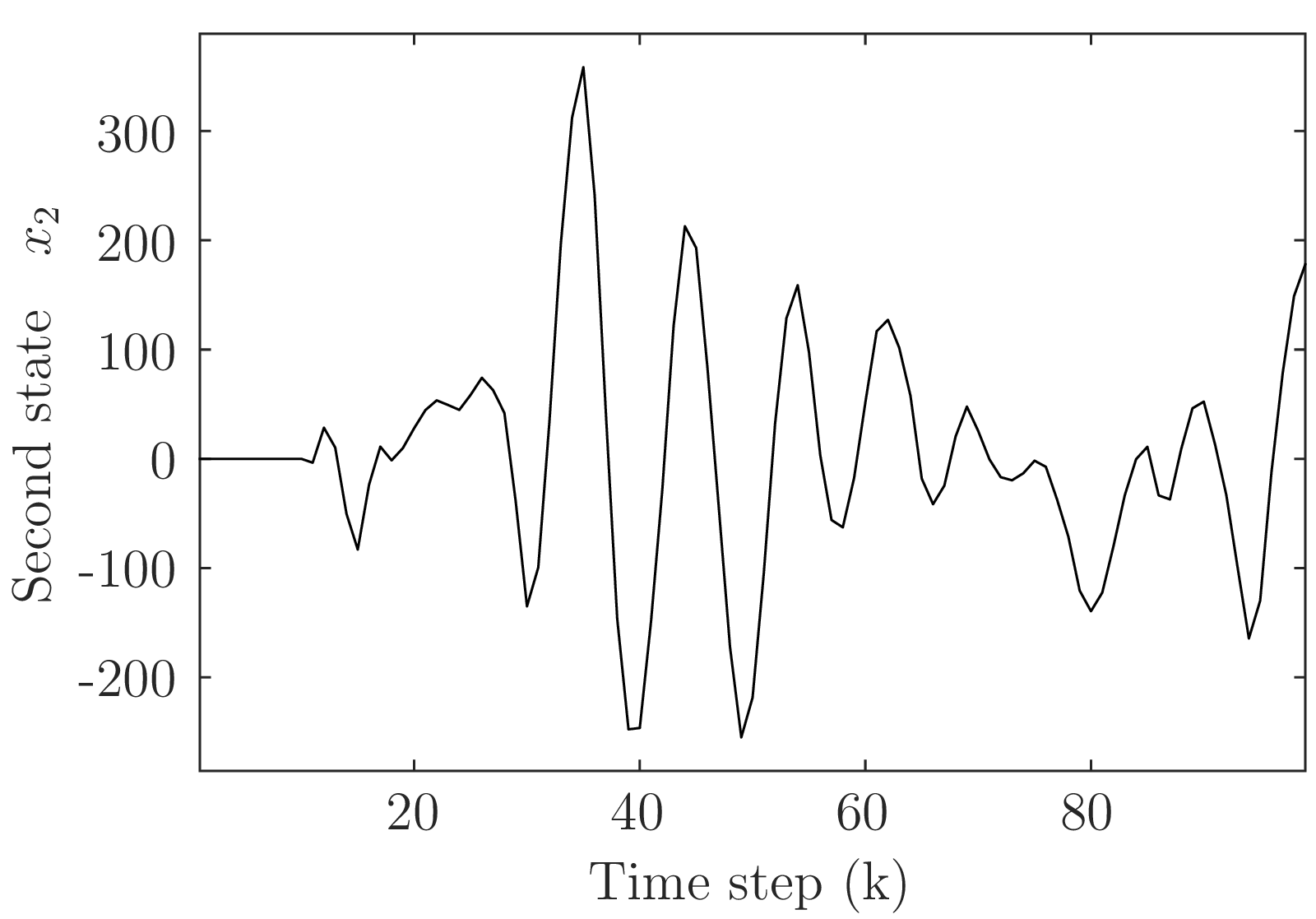}
\caption{Second state $x_{2,k}$.}
\label{fig3c}
\end{subfigure}
\par\vspace{12pt}
\begin{subfigure}[b]{2.2in}
\centering
\includegraphics[width=2.2in]{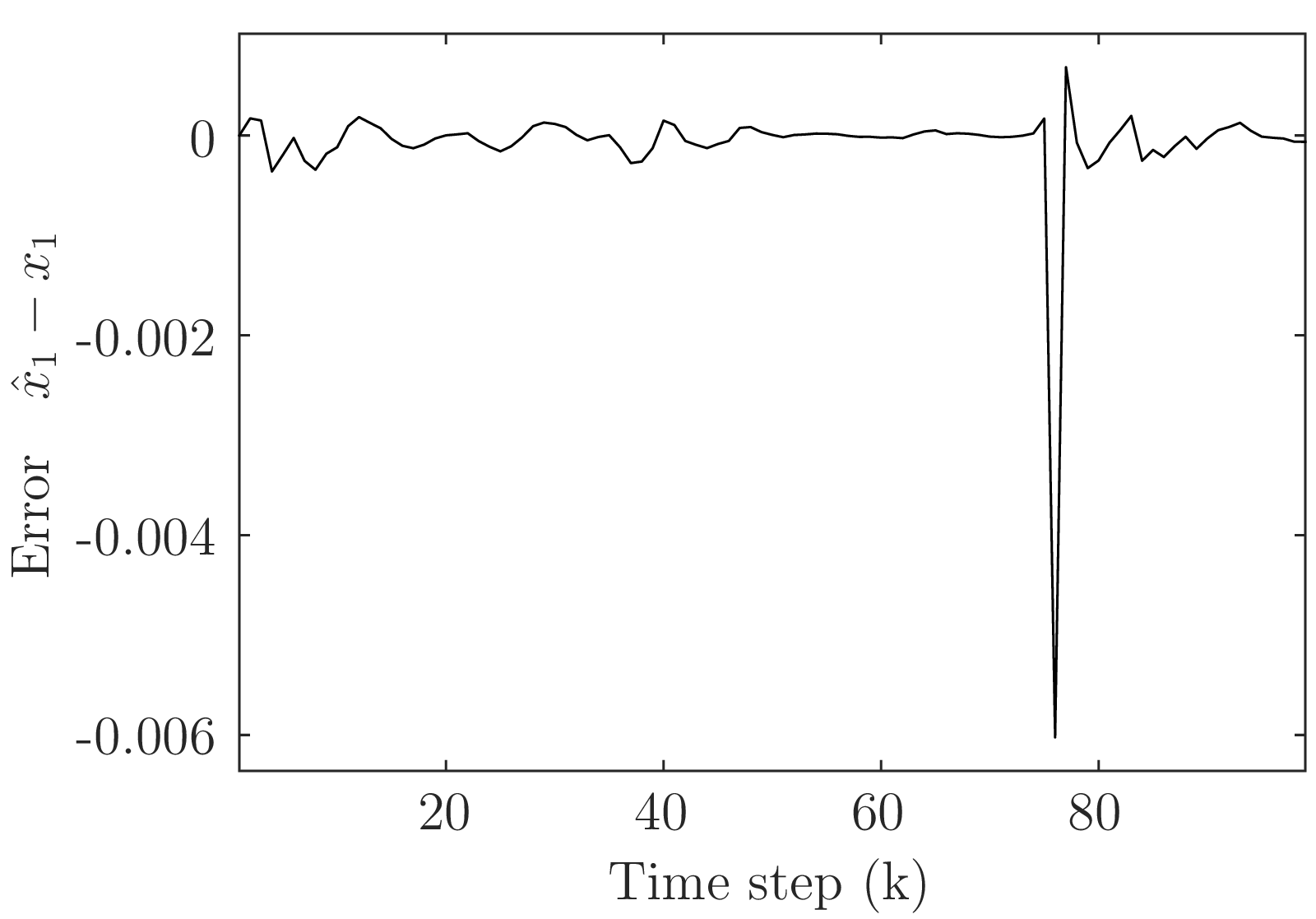}
\caption{Estimation error $x_{1,k}-\widehat{x}_{1,k}$.}\label{fig3d}
\end{subfigure}
\hfill
\begin{subfigure}[b]{2.2in}
\centering
\includegraphics[width=2.2in]{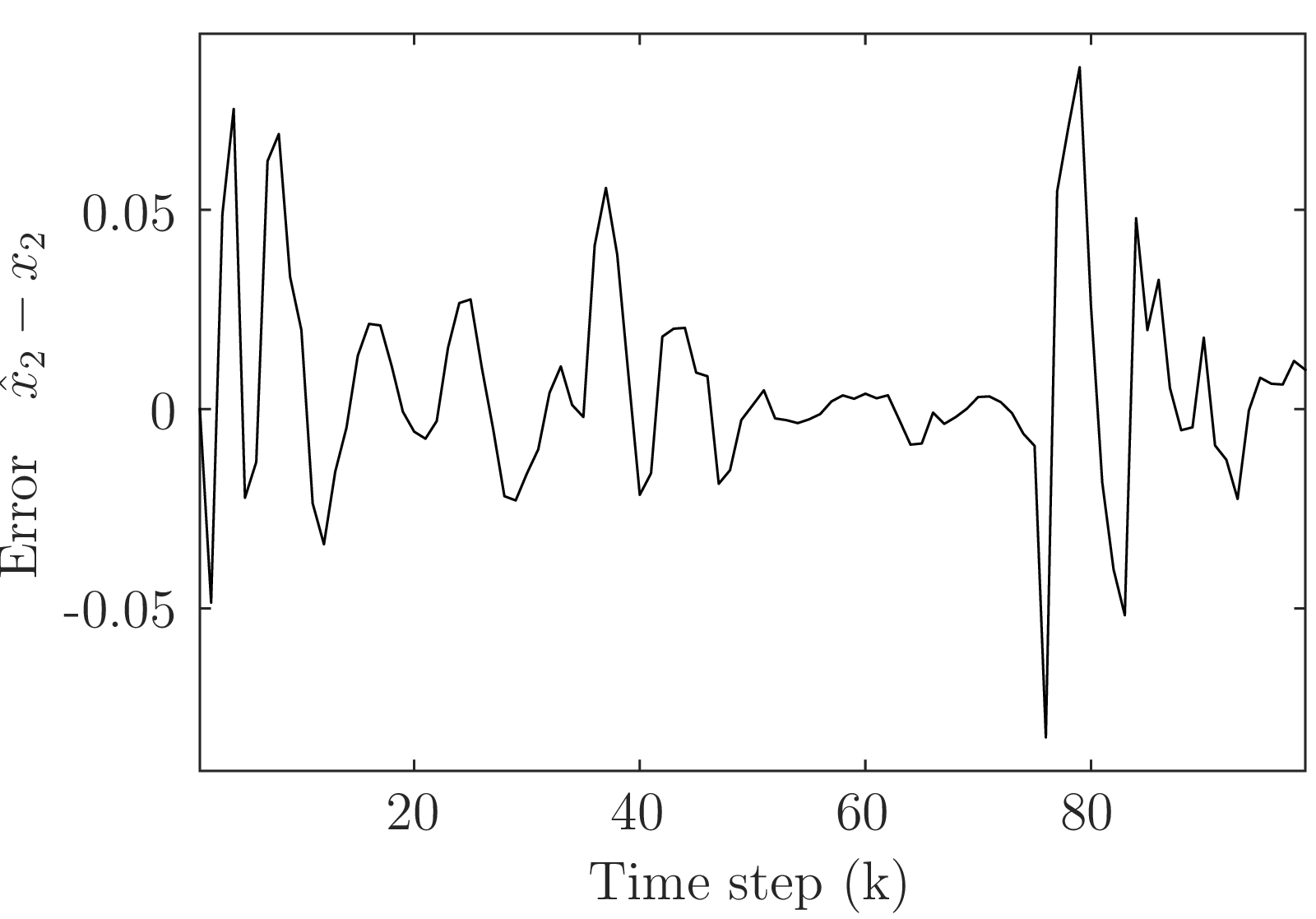}
\caption{Estimation error $x_{2,k}-\widehat{x}_{2,k}$.}\label{fig3e}
\end{subfigure}
\caption{Simulation results for the IMM Algorithm.}
\label{fig3}
\end{figure}

Due to the randomness of the mode $\theta_k$ and inputs $v_k$ and $u_k$, the simulation results are not the same for simulation trials with the same conditions. Hence, it is needed to make a statistical comparison between the simulation results of the three methods in order to draw more accurate conclusions.
For this purpose, the mode detection error percentage (\%MDE) and root mean square error for the $i$th state variable (RSME$_i$) are defined as
\begin{subequations}
\begin{align}
\mathrm{\%MDE} &= \frac{100}{N} \Big({\sum}_{k=0}^N \eta_k\Big) \\[3pt]
\eta_k &= \begin{cases}1 & \mathrm{if}\quad\theta_{k}\neq \widehat{\theta}_{k}\\[-2pt] 0 & \mathrm{if}\quad\theta_{k} =\widehat{\theta}_{k}\end{cases} \\[3pt]
\mathrm{RSME}_i &= \Big[{\sum}_{k=0}^N \big(x_{k,i} - \hat x_{k,i}\big)/N\Big]^{1/2}
\end{align}
\end{subequations}
in which $N$ is the last simulation step.

Taking the average of the above measures over 100 simulation trials, the result of comparison between Algorithm~1, Algorithm~2, and the IMM algorithm is summarized in the Table~\ref{table1}. 
According to the table, Algorithm~1 has the best mode estimation performance. On the other hand, the IMM algorithm generates a much better state estimation relying on the multiplicity of Kalman filters. Considering the fact that our main objective is to estimate the mode which stands for the packet loss occurrences, it can be concluded that the Algorithm~1 is a reasonable solution for achieving this objective.

\begin{table}
\caption{Comparison of algorithms}
\label{table1}
\begin{center}
\begin{tabular}{|c|c|c|c|}
\hline
Criterion & Algorithm 1 & Algorithm 2 & IMM algorithm \\ 
\hline
E\{\%MDE\} & $6.9$ & $13.1$ & $8.2$ \\ 
\hline
E\{RSME$_1$\} & $0.11$  & $0.15$ & $0.006$ \\
\hline
E\{RSME$_2$\} & $4.3$  & $13.2$ & $0.53$ \\
\hline
\end{tabular}
\end{center}
\end{table}

To have an insight into the reason for the weaker performance of Algorithm~2 according to the Table~\ref{table1}, the histograms of the \%MDE values among the 100 simulation trials for each of the algorithms are plotted in the Fig.~\ref{fig3}. The polts show that Algorithm~2 performs better than the IMM algorithm in many of the cases. But, there are a few cases in which the estimation based on Algorithm~2 shows a very poor performance. What happens in these cases is that it takes a large number of steps for the estimator to recover from an estimation error which results in a large number of successive mode estimation errors.

\begin{figure}
\centering
\includegraphics[width=3in]{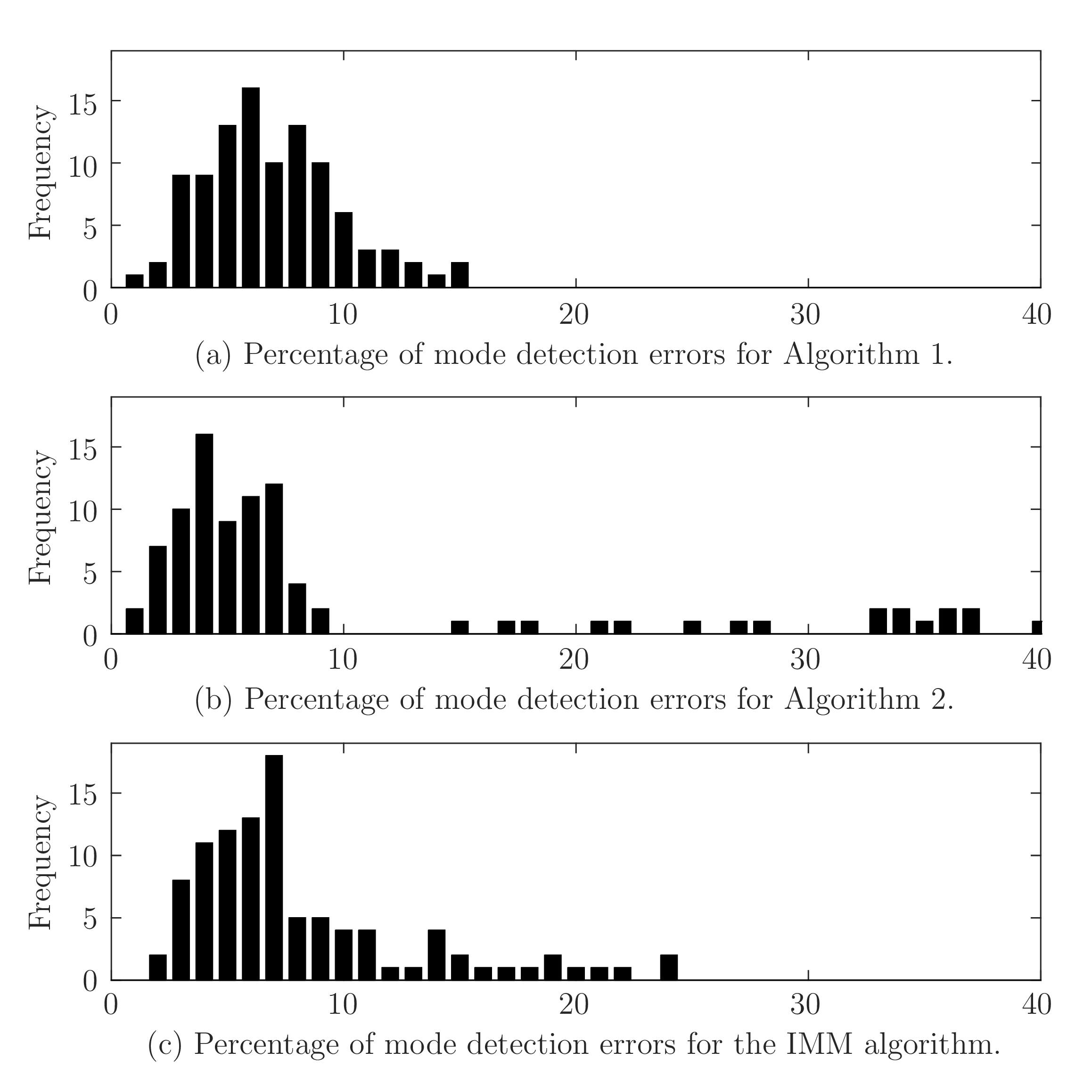}
\caption{Statistical comparison of mode detection errors for 100 simulation trials: (a) Algorithm~1, (b) Algorithm~2, (c) IMM algorithm.}
\label{fig4}
\end{figure}

\section{Conclusion} 
In this paper, two algorithms have been proposed for estimating the occurrence of packet losses represented as the mode variable of a Markovian jump system. Both of the algorithms can be used in conjunction with a single Kalman filter for simultaneous estimation of state and packet loss occurrence. 
The first algorithm is based on an input-output model of the system and is capable of being executed independently of a Kalman filter for estimation of only the packet loss occurrences. 
The second algorithm is based on the state space form and includes a Kalman filter as a component. 
Both of the algorithms have been applied to a reactor system during an example. 
It was shown that the existing multiple model estimation methods can be also applied to the simultaneous estimation problem, although there is the disadvantage that they require multiple Kalman filters. 
The performances of the proposed algorithms and the interacting multiple model estimation method (IMM) have been verified and compared through simulations. Statistical analysis of the results shows that the first algorithm has a better estimation performance for packet loss occurrences and the IMM method generates a better state estimation. Derivation of conditions for stability and boundedness of the error covariance matrix for the proposed algorithms and making improvements to the performance of the second algorithm can be considered as directions for the future research.

\section*{References}

\end{document}